\documentclass[a4paper]{article}
\usepackage{fullpage}
\usepackage[USenglish]{babel}
\usepackage{lmodern}
\usepackage[T1]{fontenc}
\usepackage[utf8]{inputenc}
\usepackage{microtype, xspace}
\usepackage{algorithm,algorithmic}
\usepackage{amsfonts,amsmath,amsthm}
\usepackage{graphicx}
\graphicspath{{./Figures/}}
\usepackage[hidelinks]{hyperref}
\usepackage{verbatim}

\newtheorem{theorem}{Theorem}
\newtheorem{lemma}[theorem]{Lemma}
\newtheorem{corollary}[theorem]{Corollary}

\newcommand\case[1]{\paragraph{Case #1.}}

\newcommand{\ar}{\ensuremath{\mathit{AR}}\xspace}
\newcommand{\at}{\ensuremath{\mathit{AT}}\xspace}
\newcommand{\iar}{\ensuremath{\mathit{IAR}}\xspace}
\renewcommand{\sp}{\ensuremath{\mathit{SP}}\xspace}
\newcommand{\spm}{\ensuremath{\mathit{SPM}}\xspace}
\newcommand{\spt}{\ensuremath{\mathit{SPT}_r}\xspace}
\newcommand{\free}{\ensuremath{\mathit{Free}}\xspace}
\renewcommand{\H}{\ensuremath{\mathcal{H}}\xspace}
\newcommand{\parent}{\ensuremath{\mathit{parent}}\xspace}

\newcommand{\etal}{\emph{et al.}\xspace}
\newcommand{\ie}{\emph{i.e.}\xspace}

\title{An Optimal Algorithm to Compute the Inverse Beacon Attraction Region}

\author{Irina Kostitsyna\thanks{TU Eindhoven, the Netherlands  \texttt{i.kostitsyna@tue.nl}}
\and
Bahram Kouhestani\thanks{Queen's University, Canada \texttt{\{kouhesta,daver\}@cs.queensu.ca}}
\and
Stefan Langerman\thanks{Directeur de Recherches du F.R.S.-FNRS., Universit\'e Libre de Bruxelles, Belgium \texttt{stefan.langerman@ulb.ac.be}}
\and
David Rappaport\footnotemark[2]}

\date{}

\begin{document}
\maketitle

\begin{abstract}
The \emph{beacon model} is a recent paradigm for guiding the
trajectory of messages or small robotic agents in complex
environments. 
A \emph{beacon} is a fixed point with an attraction pull that can move points within a given polygon. Points move greedily towards a beacon: if unobstructed, they move along a straight line to the beacon, and otherwise they slide on the edges of the polygon. The Euclidean distance from a moving point to a beacon is monotonically decreasing. A given beacon \emph{attracts} a point if the point eventually reaches the beacon.

The problem of attracting all points within a polygon with a set of beacons can be viewed as a variation of the art gallery problem. Unlike most variations, the beacon attraction has the intriguing property of being asymmetric, leading to separate definitions of \emph{attraction region} and \emph{inverse attraction region}. 
The attraction region of a beacon is the set of points that it
attracts. It is connected and can be computed in linear time for simple
polygons. 
By contrast, it is known that the inverse attraction region of a point
--- the set of beacon positions that attract it --- could have
$\Omega(n)$ disjoint connected components.

In this paper, we prove that, in spite of this, the total complexity
of the inverse attraction region of a point in a simple polygon is
linear, and present a $O(n \log n)$ time algorithm to
construct it. This improves upon the best previous algorithm which
required $O(n^3)$ time and $O(n^2)$ space.
Furthermore we prove a matching $\Omega(n\log n)$ lower bound for
this task in the algebraic computation tree model of computation, even
if the polygon is monotone.

\end{abstract}

\section{Introduction}
Consider a dense network of sensors. In practice, it is common that routing between two nodes in the network is performed by greedy geographical routing, 
where a node sends the message to its closest neighbour (by Euclidean distance) to the destination~\cite{KT10}. Depending on the geometry of the network, greedy routing may not be successful between all pairs of nodes. Thus, it is essential to determine nodes of the network for which this type of routing works. In particular, given a node in the network, it is important to compute all nodes that can successfully send a message to or receive a message from the input node. Greedy routing has been studied extensively in the literature of sensor network as a local (and therefore inexpensive) protocol for message sending.

Let $P$ be a simple polygon with $n$ vertices. A \textit{beacon} $b$ is a point in $P$ that can induce an attraction pull towards itself within $P$. 
The attraction of $b$ causes points in $P$ to move towards $b$ as long as their Euclidean distance is maximally decreasing. 
As a result, a point $p$ moves along the ray $\overrightarrow{pb}$ until it either reaches $b$ or an edge of $P$. In the latter case, $p$ slides on the edge towards $h$, the orthogonal projection of $b$ on the supporting line of the edge (Figure~\ref{fig15}). Note that among all points on the supporting line of the edge, $h$ has the minimum Euclidean distance to $b$. 

\begin{figure}[t]
\centering
\includegraphics{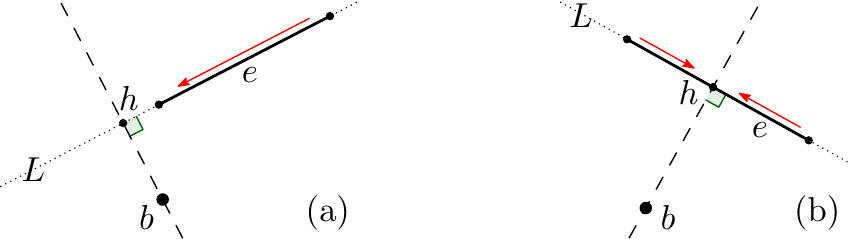}
\caption{\label{fig15} Points on an edge $e$ slide towards the orthogonal projection $h$ of the beacon on the supporting line of $e$.}
\end{figure}

We say $b$ \textit{attracts} $p$, if $p$ eventually reaches $b$.
Interestingly, beacon attraction is not symmetric. 
The \textit{attraction region} $\ar(b)$ of a beacon $b$ is the set of all points in $P$ that $b$ attracts%
\footnote{We consider the attraction region to be closed, \ie $b$ attracts all points on the boundary of  $\ar(b)$.}.
The \textit{inverse attraction region} $\iar(p)$ of a point $p$ is the set of all beacon positions in $P$ that can attract
$p$.  

The study of beacon attraction problems in a geometric domain, initiated by Biro
\etal~\cite{BGIKM}, finds its root in sensor 
networks, where the limited capabilities of sensors makes it crucial
to design simple mechanisms for guiding their motion and
communication. 
For instance, the beacon model can be used to represent the trajectory
of small robotic agents in a polygonal domain, or that of messages in
a dense sensor network.   
Using  greedy routing, the trajectory of a robot (or a message) from a sender to a
receiver closely follows the attraction trajectory of a point (the
sender) towards a beacon (the receiver). 
However, greedy routing may not be successful between all pairs of
nodes. Thus, it is essential to characterize for which pairs of nodes
of the network for which this type of routing works. In particular,
given a single node, it is important to compute the set of nodes that
it can successfully receive messages from (its attraction region), and
the set of node that it can successfully send messages to (its inverse
attraction region).

In 2013, Biro et al.~\cite{BIKM} showed that the attraction region
$\ar(b)$ of a beacon $b$ in a simple polygon $P$ is simple and
connected, and presented a linear time algorithm to
compute $\ar(b)$.

Computing the inverse attraction region has proved to be more challenging.
It is known~\cite{BIKM} that the inverse
attraction region $\iar(p)$ of a point $p$ is not necessarily
connected and can have $\Theta(n)$ connected components. 
Kouhestani~\etal~\cite{KRS15} presented an algorithm to compute
$\iar(p)$ in $O(n^3)$ time and $O(n^2)$ space. In the
special cases of monotone and terrain polygons, they showed improved
algorithms with running times $O(n \log n)$ and $O(n)$ respectively. 

In this paper, we prove that, in spite of not being connected, the inverse
attraction region $\iar(p)$ always has total
complexity\footnote{Total number of vertices and edges of all
  connected components.} $O(n)$. 
Using this fact, we present the first optimal $O(n\log n)$
time algorithm for computing $\iar(p)$ for any simple polygon $P$,
improving upon the previous best known $O(n^3)$ time algorithm.
Since this task is at the heart of other algorithms for
solving beacon routing problems, this  improves the time
complexity of several previously known algorithms such as
approximating minimum beacon paths and computing the weak attraction
region of a region~\cite{BIKM}. 

To prove the optimality of our algorithm, we show an 
$\Omega(n\log n)$ lower bound in the algebraic computation tree model
and in the bounded degree algebraic decision tree model, even in the case
when the polygon is monotone.

\subsection*{Related work}


Several geometric problems related to the beacon model have been studied in recent years.  
Biro  \etal~\cite{BGIKM} studied the minimum number of beacons
necessary to successfully route between any pair of points in a simple
$n$-gon $P$. This can be viewed as a variant of the art gallery problem,
where one wants to find the minimum number of beacons whose
attraction regions cover $P$. They proved that $\left
  \lceil\frac{n}{2}\right \rceil$ beacons are sometimes necessary and
always sufficient, 
and showed that finding a minimum cardinality set of beacons to cover a simple polygon is NP-hard. 
For polygons with holes, Biro  \etal~\cite{BGIKM2} showed that $\left \lceil\frac{n}{2}\right \rceil-h-1$ beacons are sometimes necessary and $\left \lceil\frac{n}{2}\right \rceil+h-1$ beacons are always sufficient to guard a polygon with $h$ holes. 
Combinatorial results on the use of beacons in orthogonal polygons have been studied by Bae \etal~\cite{BSV}  and by Shermer~\cite{S15}. 
Biro \etal~\cite{BIKM} presented a polynomial time algorithm for routing between two fixed points using a discrete set of candidate beacons in a simple polygon and gave a 2-approximation algorithm where the beacons are placed with no restrictions. 
Kouhestani~\etal~\cite{KRS14} give an  O($n \log n$) time algorithm
for beacon routing in a 1.5D polygonal terrain. 

Kouhestani~\etal~\cite{KRS15b} showed that the length of a successful
beacon trajectory is less than $\sqrt{2}$ times the length of a
shortest (geodesic) path. In contrast, if the polygon has internal
holes then the length of a successful beacon trajectory may be
unbounded. 





\section{Preliminaries}
A \emph{dead point} $d \neq b$ is defined as a point that remains stationary in the attraction pull of $b$. The set of all points in $P$ that eventually reach (and stay) on $d$ is called the \emph{dead region} of $b$ with respect to $d$. 
A \emph{split edge} is defined as the boundary between two dead regions, or a dead region and $\ar(b)$. In the latter case, we call the split edge a \emph{separation edge}.  

If beacon $b$ attracts a point $p$, we use the term \emph{attraction trajectory}, denoted by $\at(p,b)$, to indicate the movement path of a point $p$ from its original location to $b$. 
The attraction trajectory alternates between a straight movement
towards the beacon (a \emph{pull edge}) and a sequence of consecutive
sliding movements (\textit{slide edges}), see Figure~\ref{fig16}. 

\begin{figure}[t]
\centering
\includegraphics{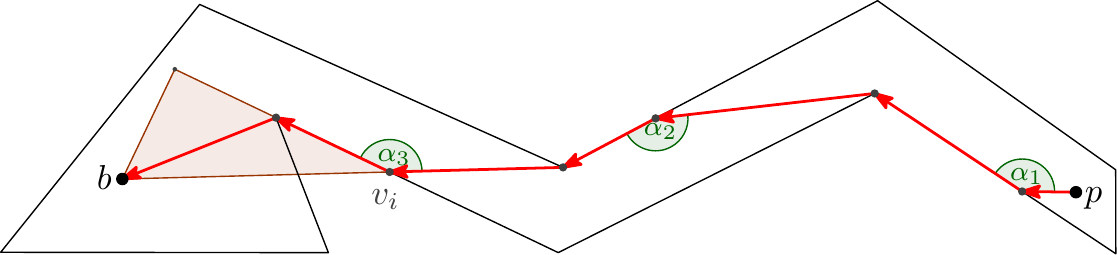}
\caption{\label{fig16} The angle between a straight movement towards the beacon and the following slide movement is always greater than $\pi / 2$.}
\end{figure}

%
\begin{lemma}
\label{lem:angle}
Consider the movement of a point $p$ in the attraction of a beacon $b$. 
Let $\alpha_i$ denote the angle between the $i$-th pull edge and the next slide edge on $\at(p,b)$ (Figure~\ref{fig16}). Then $\alpha_i$ is greater than $\pi / 2$.
\end{lemma}
\begin{proof}
Recall that a pull edge is always oriented towards $b$, and a slide edge is oriented towards the orthogonal projection of $b$ on the edge.
Consider the right triangle with vertices $b$, the orthogonal projection of $b$ on the supporting line of the slide edge, and $v_i$ the vertex common to the $i$-th pull edge and the next slide edge (the colored triangle in Figure~\ref{fig16}).  
Note that in this right  triangle, the angle of the vertex $v_i$ must be acute. Therefore, the angle of $\alpha_i$, which is the complement of $v$, is greater than $\pi / 2$. 
\end{proof}

Note that, similarly, the angle between the $i$-th pull edge and the previous slide edge is also greater than $\pi /2$.

\begin{figure}[b]
\centering
\includegraphics{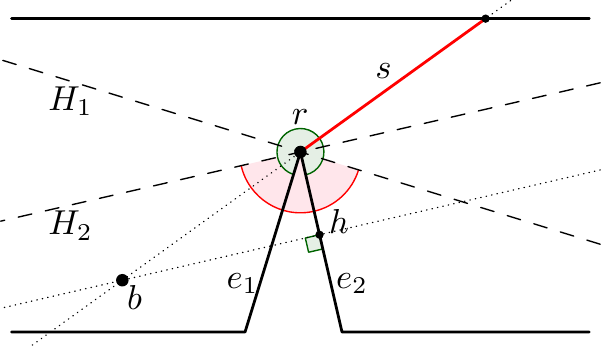}
\caption{\label{fig18} The deadwedge of $r$ is shown by the red angle.}
\end{figure}

Let $r$ be a reflex vertex of $P$ with adjacent edges $e_1$ and $e_2$. 
Let $H_1$ be the half-plane orthogonal to $e_1$ at $r$, that contains $e_1$. Let $H_2$ be the half-plane orthogonal to $e_2$ at $r$, that contains $e_2$.
The \emph{deadwedge} of $r$ (deadwedge($r$)) is defined as $H_1 \cap H_2$ (Figure~\ref{fig18}).
Let $b$ be a beacon in the deadwedge of $r$. Let $\rho$ be the ray from $r$ in the direction $\overrightarrow{br}$ and let $s$ be the line segment between $r$ and the first intersection of $\rho$ with the boundary of $P$. 
Note that in the attraction of $b$, points on different sides of $s$ have different destinations. Thus, $s$ is a split edge for $b$. We say $r$ \emph{introduces} the split edge $s$ for $b$ to show this occurrence. Kouhestani \etal~\cite{KRS15} proved the following lemma.

\begin{lemma}[Kouhestani \etal~\cite{KRS15}]
\label{lem:deadwedge}
A reflex vertex $r$ introduces a split edge for the beacon $b$ if and only if $b$ is inside the deadwedge of $r$.
\end{lemma}
%
%
Let $p$ and $q$ be two points in a polygon $P$. We use $\overline{pq}$ to denote the straight-line segment between these points.
Denote the shortest path between $p$ and $q$ in $P$ (the geodesic path) as $\sp(p,q)$. 
The union of shortest paths from $p$ to all vertices of $P$ is called
the \emph{shortest path tree} of $p$, and can be computed in linear
time~\cite{GHLST} when $P$ is a simple polygon.
In our problem, we are only interested in shortest paths from $p$ to reflex vertices of $P$. 
Therefore, we delete all convex vertices and their adjacent edges in
the shortest path tree of $p$ to obtain the  \emph{pruned shortest
  path tree} of $p$, denoted by $\spt(p)$. 

A \emph{shortest path map} for a given point $p$, denoted as
$\spm(p)$, is a subdivision of $P$ into regions such that shortest
paths from $p$ to all the points inside the same region pass through
the same set of vertices of $P$~\cite{LP}. Typically, shortest path
maps are considered in the context of polygons with holes, where the
subdivision represents grouping of the shortest paths of the same
topology, and the regions may have curved boundaries. In the case of a
simple polygon, the boundaries of $\spm(p)$ are straight-line segments
and
consist solely of the edges of $P$ and extensions of the edges of $\spt(p)$. If a triangulation of $P$ is given, it can be computed in linear time~\cite{GHLST}.

\begin{lemma}
\label{lem:moninc}
During the movement of $p$ on its beacon trajectory, the shortest path distance of $p$ away from its original location monotonically increases. 
\end{lemma}
\begin{proof}
For the sake of contradiction, let $s$ be the first point that during the movement of $p$, the shortest path away from $p$ decreases. 
Let $u$ be the last reflex vertex (before $s$) common to the attraction trajectory and the shortest path. Without loss of generality assume that the line $\overline{ub}$ is horizontal and $u$ is to the left of $b$. 
Note that on a pull edge with an arbitrary starting reflex vertex $w$, the shortest path away from $w$ monotonically increases, and therefore, as $w$ is a reflex vertex on $\sp(p,s)$, $s$ cannot be on a pull edge, and thus it is on a slide edge. 
The line $\overline{ub}$ is horizontal, therefore, a series of slide edges will result in a decrease in the shortest path towards $u$ only if during the movement of $p$ on these edges, its $x$-coordinate decreases. This results in an increase in the Euclidean distance towards $b$, which is a contradiction.  
\end{proof}

\section{The structure of inverse attraction regions}
The $O(n^3)$ time algorithm of Kouhestani \etal~\cite{KRS15} to
compute the inverse attraction region of a point $p$ in a simple
polygon $P$ 
constructs a line arrangement $A$ with quadratic complexity that
partitions $P$ into regions, such that, either all or none of the
points in a region attract $p$.  
Arrangement $A$, contains three types of lines:
\begin{enumerate}
\item Supporting lines of the deadwedge for each reflex vertex of $P$, 
\item Supporting lines of edges of $\spt(p)$,
\item Supporting lines of edges of $P$. 
\end{enumerate}

%
\begin{lemma}[Kouhestani \etal~\cite{KRS15}]
\label{lem:arrangement}
The boundary edges of $\iar(p)$ lie on the lines of arrangement $A$.
\end{lemma}

\begin{figure}[tb]
\centering
\includegraphics[scale=0.9]{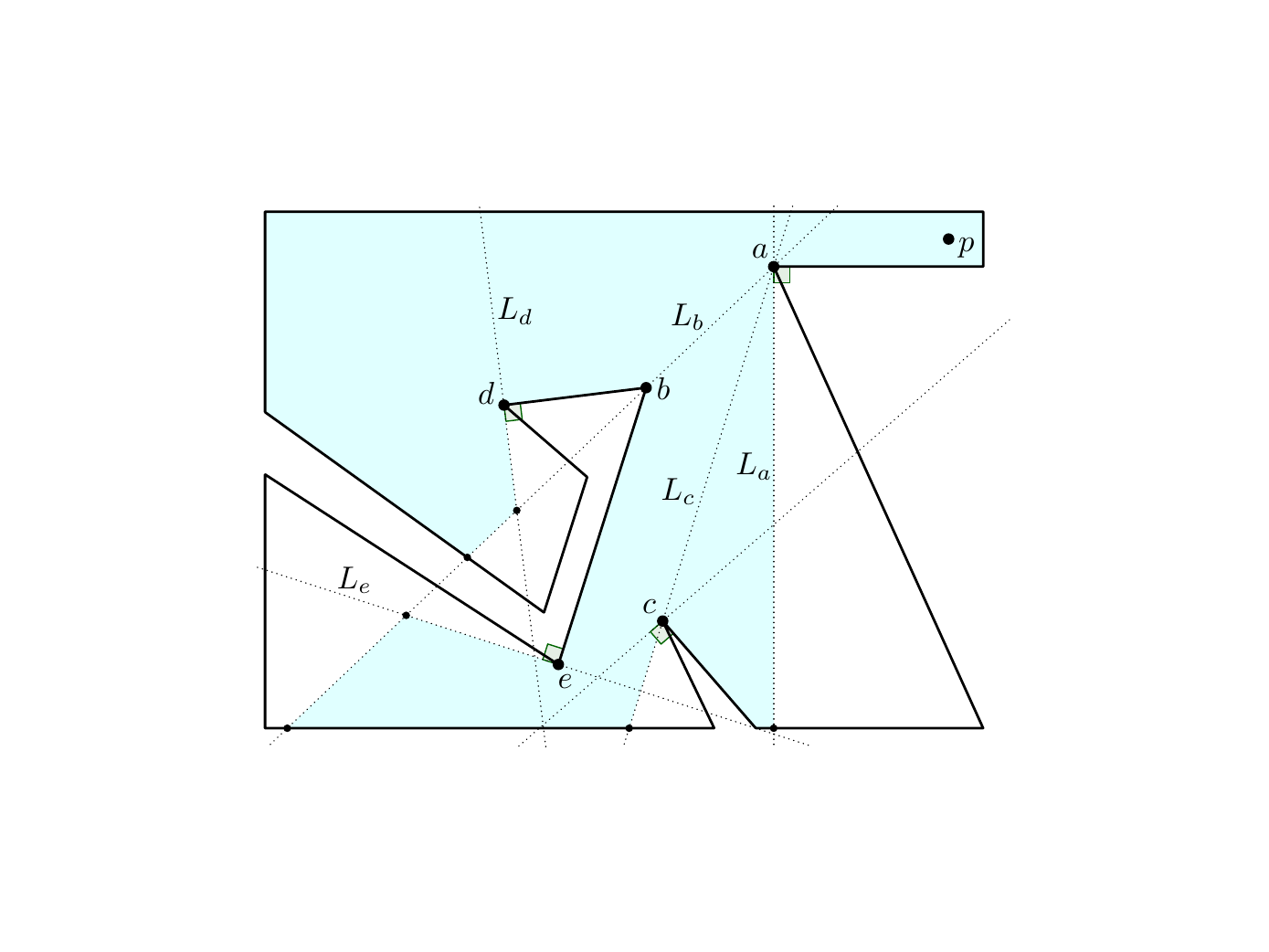}
\caption{\label{fig38} An example of an inverse attraction region with effective associated lines to each reflex vertex. Points in the colored region attract $p$. Here $L_a$, $L_b$, $L_c$, $L_d$ and $L_e$ are respectively the associated lines of the reflex vertices $a$, $b$, $c$, $d$ and $e$. }
\end{figure}

\noindent Let $\overline{uv}$ be an edge of $\spt(p)$, where $u = \parent(v)$. We associate three lines of the arrangement $A$ to $\overline{uv}$: supporting line of $\overline{uv}$ and the two supporting lines of the deadwedge of $v$.
By focusing on the edge $\overline{uv}$, we study the local effect of the reflex vertex $v$ on $\iar(p)$, and we show that:
\begin{enumerate}
\item Exactly one of the associated lines to $\overline{uv}$ may contribute to the boundary of $\iar(p)$. We call this line the \emph{effective associated line} of  $\overline{uv}$ (Figure~\ref{fig38}). 
\item The effect of $v$ on the inverse attraction region can be represented by at most two half-planes, which we call the \emph{constraining half-planes} of $\overline{uv}$. These half-planes are bounded by the effective associated line of $\overline{uv}$.
\item Each constraining half-plane has a \textit{domain}, which is a subpolygon of $P$ that it affects. The points of the constraining half-plane that are inside the domain subpolygon cannot attract $p$ (see the next section).
\end{enumerate}

\noindent Our algorithm to compute the inverse attraction region uses $\spm(p)$. For each region of $\spm(p)$, 
we compute the set of constraining half-planes with their domain subpolygons containing the region.
Then, we discard points of the region that cannot attract $p$ by locating points which belong to at least one of these constraining half-planes.

%
%
%
%
\subsection*{Constraining half-planes}\label{sec:3.1}
Let $\overline{uv}$ be an edge of $\spt(p)$, where $u = \parent(v)$. We extend $\overline{uv}$ from $u$ until we reach $w$, the first intersection with the boundary of $P$ (Figure~\ref{fig56}). Segment $\overline{uw}$ partitions $P$ into two subpolygons. Let $P_p$ be the subpolygon that contains $p$. 
Any path from $p$ to any point in $P\setminus P_p$ passes through $\overline{uw}$. Thus a beacon outside of $P_p$ that attracts $p$, must be able to attract at least one point on the line segment $\overline{uw}$. 
In order to determine the local attraction behaviour caused by the vertex $v$, and to find the effective line associated to $\overline{uv}$, we focus on the attraction pull on the points of $\overline{uw}$ (particularly the vertex $u$) rather than $p$. 
By doing so we detect points that cannot attract $u$, or any point on $\overline{uw}$, and mark them as points that cannot attract $p$.  
In other words, for each edge $\overline{uv} \in \spt(p)$ we detect a set of points in $P$ that cannot attract $u$ locally due to $v$. The attraction of these beacons either causes $u$ to move to a wrong subpolygon,  or their attraction cannot move $u$ past $v$ (see the following two cases for details). 
Later in Theorem~\ref{thm:1}, we show that this suffices to detect all points that cannot attract $p$.

\begin{figure}[b]
\centering
\includegraphics{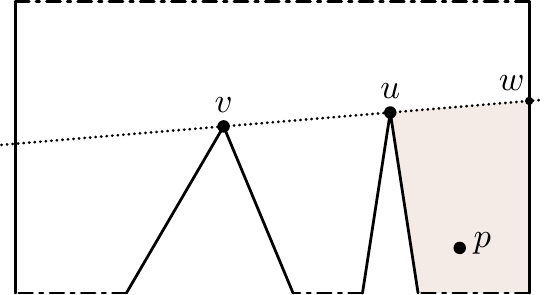}
\caption{\label{fig56} The possible locations of $p$ with respect to the shortest path edge $\overline{uv}$.
}
\end{figure}
%
%
%
%

Let $e_{1}$ and $e_{2}$ be the edges incident to $v$.
Let $H_{1}$ be the half-plane, defined by a line orthogonal to $e_{1}$ passing through $v$, which contains $e_{1}$, and let $H_{2}$ be the half-plane, defined by a line orthogonal to $e_{2}$ passing through $v$, which contains $e_{2}$.
Depending on whether $u$ is in $H_{1} \cup H_{2}$, we consider two cases:
%
%


\case{1} Vertex $u$ is not in $H_{1} \cup H_{2}$ (Figure~\ref{fig37}). 
%
We show that in this case the supporting line of $\overline{uv}$ is the only line associated to $v$ that may contribute to the boundary of $\iar(p)$, \ie it is the effective line associated to $\overline{uv}$. 
Let $q$ be an arbitrary point on the open edge $e_1$. 
As $u$ is not in $H_{1} \cup H_{2}$, the angle between the line segments $\overline{uq}$ and $\overline{qv}$ is less than $\pi /2$. 
Consider an arbitrary attraction trajectory that moves $u$ straight towards $q$. By Lemma~\ref{lem:angle}, any slide movement of this attraction trajectory on the edge $e_1$ moves away from $v$. 
Now consider $q$ to be on the edge $e_2$. Similarly any slide on the edge $e_2$ moves away from $v$. 
Thus, an attraction trajectory of $u$ can cross the line segment $\overline{uv}$ only once (the same holds for any other point on the line segment $\overline{uw}$). Note that this crossing movement  happens via a pull edge. We use this observation to detect a set of points that do not attract $u$ and thus do not attract $p$. 


Now consider the supporting line $L$ of the edge $\overline{uv}$. As $u$ is not in $H_{1} \cup H_{2}$, $L$ partitions the plane into two half-planes $L_{1}$ containing the edge $e_1$, and $L_{2}$ containing the edge $e_2$. 
Without loss of generality, assume that the parent of $u$ in $\spt(p)$ lies inside $L_{2}$ (refer to Figure~\ref{fig37}). 
Recall that $\overline{uw}$ partitions $P$ into two subpolygons, and $P_p$ is the subpolygon containing $p$. 
We define subpolygons $P_{1}$ and $P_{2}$ as follows.
Let $\rho_1$ be the ray originating at $v$, perpendicular to $L$ in $L_{1}$, and let $z_1$ be the first intersection point of $\rho_1$ with the boundary of $P$. Define $P_{1}$ as the subpolygon of $P$ induced by $\overline{vz_1}$ that contains the edge $e_1$.
Similarly, let $\rho_2$ be the ray originating at $v$, perpendicular to $L$ inside $L_{2}$, and let $z_2$ be the first intersection point of $\rho_2$ with the boundary of $P$. Define $P_{2}$ as the subpolygon of $P$ induced by $\overline{vz_2}$ that contains the edge $e_2$. 
%
\begin{lemma}
\label{lem02}
No point in $P_{1} \cap L_{2}$ can attract $p$. 
\end{lemma}
\begin{proof}
Without loss of generality assume the position in Figure~\ref{fig37}. 
Consider a beacon $b_1$ in $P_{1} \cap L_{2}$. If $b_1$ is on or above the ray $Z_2$ then in the attraction of $b_1$ a point on $\rho_1$ will move away from $P_{1}$. Therefore, in this case $b_1$ does not attract any point outside of $P_{1}$ including $p$. 
Now if $b_1$ is below the ray $Z_2$ then any straight movement from $u$ to $b_1$ is towards the edge $e_2$ and therefore in the attraction of $b_1$, no point on $\overline{uw}$ can enter $P_{1}$ directly without sliding on $e_2$. 
As we explained earlier, any slide on the edge $e_2$ moves away from $v$, and therefore, $b_1$ cannot attract $u$. Similarly $b_1$ cannot attract any point on $\overline{uw}$.
As the attraction trajectory of $p$ towards $b_1$ must pass through $\overline{uw}$, $b_1$ cannot attract $p$.
\end{proof}
%
%

\begin{lemma}
\label{lem03}
No point in $P_{2} \cap L_{1}$ can attract $p$. 
\end{lemma}
\begin{proof}
Without loss of generality assume the position in Figure~\ref{fig37}. 
Consider a beacon $b_2$ in $P_{2} \cap L_{1}$. If $b_2$ is on or above the ray $Z_1$ then in the attraction of $b_2$ a point on $\rho_2$ will move away from $P_{2}$. Therefore, in this case $b_2$ does not attract any point outside of $P_{2}$ including $p$. 
Now if $b_2$ is below the ray $Z_1$ then, in the attraction of $b_2$, no points on $\overline{uw}$ can cross $\overline{uv}$ without sliding on $e_1$. 
As we explained earlier, any slide on the edge $e_1$ moves away from $v$.
Therefore, $b_2$ cannot attract $u$ or any point on $\overline{uw}$, and so it cannot attract $p$.
\end{proof}

\noindent In summary, in case 1, 
the effect of $\overline{uv}$ is expressed by two  half-planes: $L_{2}$, affecting the subpolygon $P_{1}$, and $L_{1}$, affecting the subpolygon $P_{2}$. We call $L_{1}$ and $L_{2}$ the \textit{constraining half-planes} of $\overline{uv}$, and we call $P_{1}$ and $P_{2}$ the \textit{domain} of the constraining half-planes $L_{2}$ and $L_{1}$, respectively. 
Furthermore, we call $P_{1} \cap L_{2}$ and $P_{2} \cap L_{1}$ the \textit{constraining regions} of $\overline{uv}$. 
Later we show that $L$ is the only effective line associated to $\overline{uv}$. 

\begin{figure}[t]
\begin{minipage}[t]{0.42\textwidth}
\centering
\includegraphics{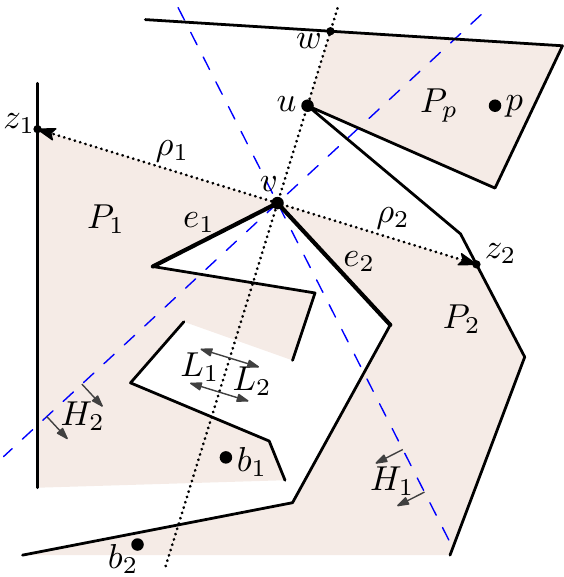}
\caption{\label{fig37} Vertex $u\not\in H_{1} \cup H_{2}$. Subpolygon $P_{2}$ is the domain of the constraining half-plane $H_{1}$, and $P_{1}$ is the domain of the constraining half-plane $H_{2}$.}
\end{minipage}
\hfill
\begin{minipage}[t]{0.53\textwidth}
\centering
\includegraphics{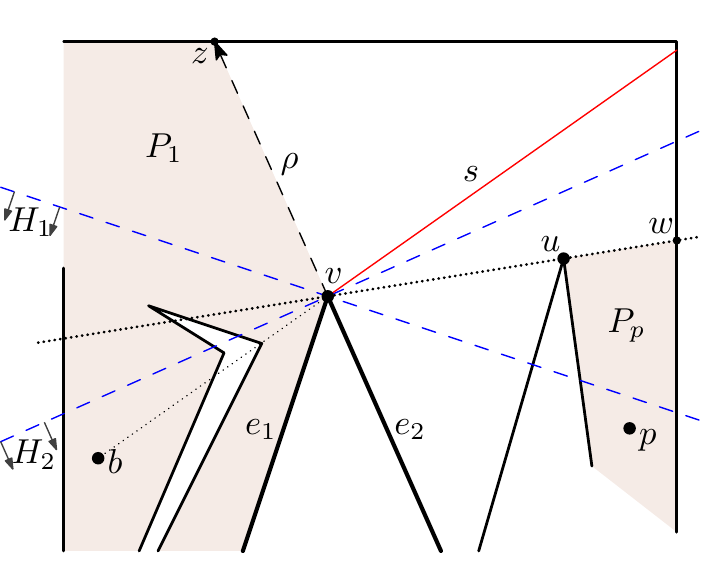}
\caption{\label{fig36} Vertex $u\in H_{1} \cup H_{2}$. Subpolygon $P_{1}$ is the domain of the constraining half-plane $H_{2}$.}
\end{minipage}
\end{figure}

\case{2} Vertex $u$ is in $H_{1} \cup H_{2}$ (refer to Figure~\ref{fig36}).
Without loss of generality assume $u$ can see part of the edge $e_2$.
Similar to the previous case, we define the subpolygon $P_p$;
let $w$ be the first intersection of the ray $\overrightarrow{vu}$ with the boundary of $P$.
Note that $\overline{uw}$ partitions $P$ into two subpolygons. Let $P_p$ be the subpolygon containing $p$. 
Now let $\rho$ be the ray originating at $v$, along the extension of edge $e_2$. Let $z$ be the first intersection of $\rho$ with the boundary of $P$. We use $P_{1}$ to denote the subpolygon induced by $\overline{vz}$ that contains $e_1$. We detect points in $P_{1}$ that cannot move $u$ (past $v$) into $P_{1}$.
%
%
%
%
\begin{lemma}
\label{lem01}
No point in $P_{1} \cap H_{2}$ can attract $p$.
\end{lemma}
\begin{proof}
Without loss of generality assume the position in Figure~\ref{fig36}. 
Consider a beacon $b$ in $P_{1} \cap H_{2}$. If $b$ is on or to the right of the ray $\rho$ then in the attraction of $b$ a point on $\rho$ will move away from $P_{1}$. Therefore, in this case $b$ does not attract any point outside of $P_{1}$ including $p$. 
Now assume $b$ is to the left of the ray $\rho$.
As $b$ is in $H_{2}$ the orthogonal projection of $b$ on the supporting line of the edge $e_2$ also lies in $H_{2}$. Therefore, as $b$ is in $P_{1}$, it does not attract any point on the open edge $e_2$.
Consider the attraction trajectory of $u$ with respect to $b$.
As $b$ is below the supporting line of $\overline{uv}$, $u$ cannot enter $P_{1}$ via a pull edge.  
In addition, $u$ cannot slide on $e_2$ to reach $v$.
Therefore $b$ cannot attract $u$ (or similarly any point on $\overline{uw}$). Thus it does not attract $p$.
\end{proof}
%
%
%
%
\noindent In summary, in case 2,  the effect of $\overline{uv}$ on $\iar(p)$ can be expressed by the half-plane $H_{2}$.
We call $H_{2}$ the \textit{constraining half-plane} of $\overline{uv}$, $P_{1}$ the \textit{domain} of $H_{2}$ and we call $P_{1} \cap H_{2}$ the \textit{constraining region} of $\overline{uv}$. 
Later we show that the supporting line of $H_{2}$ is the only effective  line associated to $v$. 
%
%
%
By combining these two cases, we prove the following theorem.
\begin{theorem}
\label{thm:1}
A beacon $b$ can attract a point $p$ if and only if $b$ is not in a constraining region of any edge of $\spt(p)$.
\end{theorem}
\begin{proof}
By Lemmas \ref{lem02}, \ref{lem03} and \ref{lem01}, if $b$ is in the constraining region of an edge $\overline{uv} \in \spt(p)$ then it does not attract $p$. 


Now let $b$ be a point that cannot attract $p$. We will show that $b$ is in the constraining region of at least one edge of $\spt(p)$. 
Let $s$ be the separation edge of $\ar(b)$ such that $b$ and $p$ are in different subpolygons induced by $s$ (see, for example, Figure~\ref{fig36}). Note that as the attraction region of a beacon is connected~\cite{B13}, there is exactly one such separation edge. Let $v$ be the reflex vertex that introduces $s$ and let $u$ be the parent of $v$ in $\spt(p)$. By Lemma~\ref{lem:deadwedge}, $b$ is in the deadwedge of $v$. In addition, as the attraction region of a beacon is connected, $b$ attracts $v$. We claim that $b$ is  in a constraining region of the edge $\overline{uv} \in \spt(p)$.
First, we show that $b$ cannot attract $u$.
Consider $\sp(p,u)$, the shortest path from $p$ to $u$. If $\sp(p,u)$ crosses $s$ at some point $q$ then $u$ cannot be the parent of $v$ in $\spt(p)$, because we can reach $v$ with a shorter path by following $\sp(p,u)$ from $p$ to $q$ and then reaching $v$ from $q$. 
Therefore, $\sp(p,u)$ does not cross $s$, so $p$ and $u$ are in the same subpolygon of $P$ induced by $s$. As $b$ does not attract $p$, we conclude that $b$ does not attract $u$.

Now depending on the relative position of $u$ and $v$ (whether $u$ is in $H_{1} \cup H_{2}$ or not), we consider two cases.  
We show that in each case, $b$ is in a constraining region of $\overline{uv}$.

\case{1} Vertex $u$ is not in $H_{1} \cup H_{2}$ (refer to Figure~\ref{fig37}). 
Let $L$ be the supporting line of $\overline{uv}$, and similar to the previous case analysis let $L_{1}$ and $L_{2}$ be the constraining half-planes, and let $P_{1}$ and $P_{2}$ be the domains of $L_{2}$ and $L_{1}$, respectively. 
Without loss of generality, assume that $b$ is in the half-plane $L_{2}$. 
We show that then $b$ belongs to $P_{1}$.

%

As $b \in L_{2}$, the separation edge $s$ extends from $v$ into $L_{1}$, \ie $s\in L_{1}$. Then the point $p$ and subpolygon $P_{2}$ lie on one side of $s$, and subpolygon $P_{1}$ lies on the other side of $s$. As beacon $b$ does not attract $p$, the point $p$ and the beacon $b$ lie on different sides of $s$, and thus the beacon $b$ and subpolygon $P_{1}$ lie on the same side of $s$.

We will show now that indeed $b\in P_{1}$. Beacon $b$ attracts $v$ and is in the deadwedge of $v$. Thus, in the attraction of $b$, $v$ will enter $P_{1}$ via a slide move. We claim that $v$ cannot leave $P_{1}$ afterwards. 
Consider the supporting line of $\rho_1$ which is a line orthogonal to $\overline{uv}$ at $v$. As $u$ is not in $H_{1} \cup H_{2}$, and the deadwedge of $v$ is equal to $H_{1} \cap H_{2}$, the deadwedge of $v$ completely lies to one side of the supporting line. Therefore, in the attraction of $v$ by any beacon inside the deadwedge of $v$, any point $q\not=v$ on $\overline{vz_{1}}$ moves straight towards the beacon along the ray $\overrightarrow{q b}$. In other words, in the attraction pull of $b$ no point inside $P_{1}$ can leave $P_{1}$. Therefore, $b \in P_{1}$ and thus $b \in P_{1} \cap L_{2}$. By definition, $b$ belongs to a constraining region of $\overline{uv}$.

\case{2} Vertex $u$ is in $H_{1} \cup H_{2}$ (refer to Figure~\ref{fig36}). 
Without loss of generality let $u\in H_{2}$. Consider the separation edge $s$. As the beacon $b$ does not attract $u$, they lie on the opposite sides of $s$. 
As $b$ is in the deadwedge of $v$, it is also in $H_{2}$, the constraining half-plane of $\overline{uv}$. Similar to the previous case, as $b$ attracts $v$, $\at(v,b)$ never crosses $\rho$ to leave $P_{1}$ and therefore, $b$ is in $P_{1}$. 
Thus, $b \in P_{1} \cap H_{2}$ and it belongs to the constraining region of $\overline{uv}$.
\end{proof}
%
%
%
\begin{corollary}
\label{cor:effective}
Consider the edge $\overline{uv} \in \spt(p)$. If $u$ is not in $H_{1} \cup H_{2}$ (case 1), then among three associated lines to $\overline{uv}$ only the supporting line of $\overline{uv}$ may contribute to the boundary of $\iar(p)$.
If $u$ is  in $H_{1} \cup H_{2}$ (case 2), then among three associated lines to $\overline{uv}$ only the supporting line of $H_{2}$ may contribute to the boundary of $\iar(p)$, where $H_{2}$ is the half-plane orthogonal to the incident edge of $v$ that $u$ can partially see.
\end{corollary}
\section{The complexity of the inverse attraction region}
In this section we show that in a simple polygon $P$ the complexity of $\iar(p)$  is linear with respect to the size of $P$. 

We classify the vertices of the inverse attraction region into two groups: 1) vertices that are on the boundary of $P$, and 2) internal vertices. We claim that there are at most a linear number of vertices in each group. Throughout this section, without loss of generality, we assume that no two constraining half-planes of different edges of the shortest path tree are co-linear. Note that we can reach such a configuration with a small perturbation of the input points, which may just add to the number of vertices of $\iar(p)$.

Biro~\cite{B13} showed that the inverse attraction region of a point in a simple polygon $P$ is convex with respect to $P$.\footnote{A subpolygon $Q \subseteq P$ is \textit{convex with respect to} the polygon $P$ if the line segment connecting two arbitrary points of $Q$ either completely lies in $Q$ or intersects $P$.} 
Therefore, we have at most two vertices of $\iar(p)$ on each edge of $P$, and thus there are at most a linear number of vertices in the first group.

We use the following property of the attraction trajectory to count the number of vertices in group 2. 
%

\begin{lemma}
\label{lem:ober}
Let $L$ be the effective line associated to the edge $\overline{uv} \in \spt(p)$, where $u = \parent(v)$. Let $b$ be a beacon on 
$L \cap deadwedge(v)$ 
that attracts $p$. Then the attraction trajectory of $p$ passes through both $u$ and $v$. 
\end{lemma}
\begin{proof}
Consider the two cases in Section 3.1 (Figure~\ref{fig37} and Figure~\ref{fig36}). 
Recall that $w$ is the first intersection of the vector $\overrightarrow{vu}$ with the boundary of $P$, and
cutting through the line segment $\overline{uw}$ partitions $P$ into two subpolygons such that $b$ and $p$ are in different subpolygons. 
And thus $\at(p,b)$ passes through $\overline{uw}$.
In case 1 (Figure~\ref{fig37}), as $L$ is  the supporting line of $\overline{uv}$, in the attraction pull of $b$, a point on $\overline{uw}$ moves along the line segment $\overline{vw}$ and meets both $u$ and $v$.
In case 2 (Figure~\ref{fig36}), as $b$ is on $L$, it is below the supporting line of $\overline{vw}$ and therefore, $\at(p,b)$ can pass $\overline{uw}$ and $\overline{uv}$ only through $u$ and $v$ via a slide edge, respectively. 
\end{proof}

\noindent Next we define an ordering on the constraining half-planes. Let $C$ be a constraining half-plane of the edge $\overline{uv} \in \spt(p)$ ($u = \parent(v)$), and let $C'$ be a constraining half-plane of the edge $\overline{u'v'} \in \spt(p)$ ($u' = \parent(v')$). 
We say $C \leq C'$ if and only if $\left |  \sp(p,v) \right | \leq \left | \sp(p,v') \right |$ (refer to Figure~\ref{fig53}). 
\begin{figure}[t]
\begin{minipage}[t]{0.56\textwidth}
\centering
\includegraphics{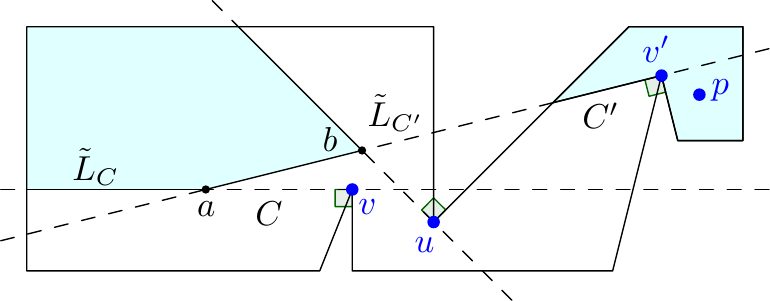}
\caption{\label{fig53} The charging scheme: charge vertex $a$ to the constraining half-plane $C$ of the vertex $v$. The inverse attraction region of $p$ is the shaded region.}
\end{minipage}
\hfill
\begin{minipage}[t]{0.4\textwidth}
\centering
\includegraphics[scale=0.9]{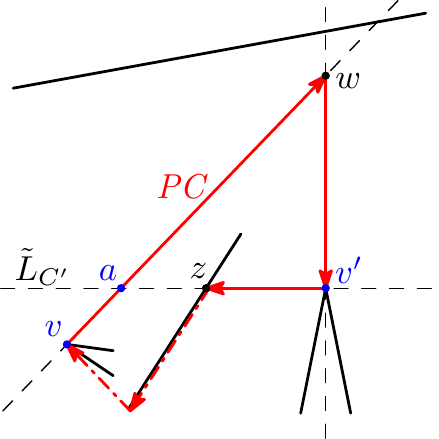}
\caption{\label{fig57} Assumptions of Lemma~\ref{lem:charge}. The chain $PC$ is shown in red.}
\end{minipage}
\end{figure}

We use a charging scheme to count the number of internal vertices. An internal vertex resulting from the intersection of two constraining half-planes $C$ and $C'$ is charged to $C'$ if $C \leq C'$, otherwise it is charged to $C$. In the remaining of this section, we show that each constraining half-plane is charged at most twice. 
Let $P_C$ and $P_C'$ denote the constraining regions related to $C$ and $C'$, respectively. 
And let $L_C$ and $L_{C'}$ denote the supporting lines of $C$ and $C'$, respectively. 
In the previous section we showed that the line segments $L_C \cap P_C$ are the only parts of $L_C$ that may contribute to the boundary of $\iar(p)$. Let $s \in L_C \cap P_C$ be a segment outside of the deadwedge of $v$. 
The next lemma shows that $s$ does not appear on the boundary of $\iar(p)$, and we can ignore $s$ when counting the internal vertices of $\iar(p)$.

\begin{lemma}
\label{lem:no-deadwedge}
Let $s \in L_C \cap P_C$ be a segment outside of the deadwedge of $v$. Then $s$ (or a part of $s$ with a non-zero length) does not appear on the boundary of $\iar(p)$. 
\end{lemma}
\begin{proof}
By Lemma~\ref{lem:deadwedge}, vertex $v$ does not introduce a split edge for any point on $s$, and thus $v$ (or the edge $\overline{uv}$) does not have an effect on the destination of the points on different sides of $s$ in the attraction pull of $b$. As we assume that no two constraining half-planes of different edges of the shortest path tree are co-linear, no constraining half-plane of any other vertex is co-linear with $s$, and the lemma follows. 
\end{proof}

\noindent We define $\tilde{L}_{C} = L_C \cap P_C ~\cap$ deadwedge($v$) and $\tilde{L}_{C'} = L_{C'} \cap P_{C'} ~\cap$ deadwedge($v'$). 
By Lemma~\ref{lem:no-deadwedge}, $\tilde{L}_{C}$ and $\tilde{L}_{C'}$ are the subset of $L_C$ and $L_{C'}$  that may appear on the boundary of $\iar(p)$, therefore, the intersection points of all  $\tilde{L}_{C}$ and $\tilde{L}_{C'}$ are the only possible locations for internal vertices of $\iar(p)$. 
Consider an internal vertex $a$ resulting from the intersection of $\tilde{L}_{C}$ and $\tilde{L}_{C'}$.

\begin{lemma}
\label{lem:charge}
Let $a = \tilde{L}_{C} \cap \tilde{L}_{C'}$ be an internal vertex of $\iar(p)$ and let $C' \leq C$ (Figure~\ref{fig53}). Then all points on $\tilde{L}_{C}$ are in the domain of $C'$.
\end{lemma}
\begin{proof}
Consider a beacon $a$. By Lemma~\ref{lem:ober}, $\at(p,a)$ passes through  both $v'$ and $v$. As $C' \leq C$, we have that $\left |  \sp(p,v') \right | \leq \left | \sp(p,v) \right |$, and therefore, by Lemma~\ref{lem:moninc}, $\at(p,a)$ reaches $v'$ before $v$. 
Recall from the proof of Theorem~\ref{thm:1} that $\at(a,v')$ does not leave the domain of $C'$, and thus $v$ belongs to the domain of $C'$. 
Without loss of generality, assume that  $\tilde{L}_{C'}$ is horizontal and the constraining half-plane of $C'$ is below this horizontal line and $a$ is to the left of $v'$ (Figure~\ref{fig57}).

For the sake of contradiction assume $\tilde{L}_{C} \not\subset P_{C'}$, then $\tilde{L}_{C}$ must intersect the boundary of the domain of $C'$. This happens only if $v$ lies below the supporting line of $\tilde{L}_{C'}$ and to the left of $a$ (Figure~\ref{fig57}). 
Let $z$ be the closest point on the line segment $\overline{v'a}$ to $a$ that $\at(a,v')$ passes through.
Consider the polygonal chain $PC = \at(v',a) \cup \overline{aw} \cup \overline{wv'}$. The chain $PC$ does not cross any edges of $P$, and at the same time, there are points on $P$ inside and outside of this chain; adjacent vertices of $v'$ are outside of $PC$ and the point $z$ (and at least one adjacent vertex to $z$) is inside of $PC$.  This contradicts the simplicity of $P$.  
\end{proof}
%

\noindent We charge $a$ to $C$ if $C' \leq C$, otherwise we charge it to $C'$. 
Assume $a$ is charged to $C$.  
By Lemma~\ref{lem:charge}, all points on $\tilde{L}_{C}$ to one side of $a$ belong to the domain of $C'$ and therefore are in $C'$. Thus, $C$ cannot contribute any other internal vertices to this side of $a$. This implies that  $C$ can be charged at most twice (once from  each end) and as there are a linear number of constraining half-planes, we have at most a linear number of vertices of group 2, and we have the following theorem. 

\begin{theorem}
The inverse attraction region of a point $p$ has linear complexity in a simple polygon.
\end{theorem} 
%

\begin{figure}[t]
\centering
\includegraphics{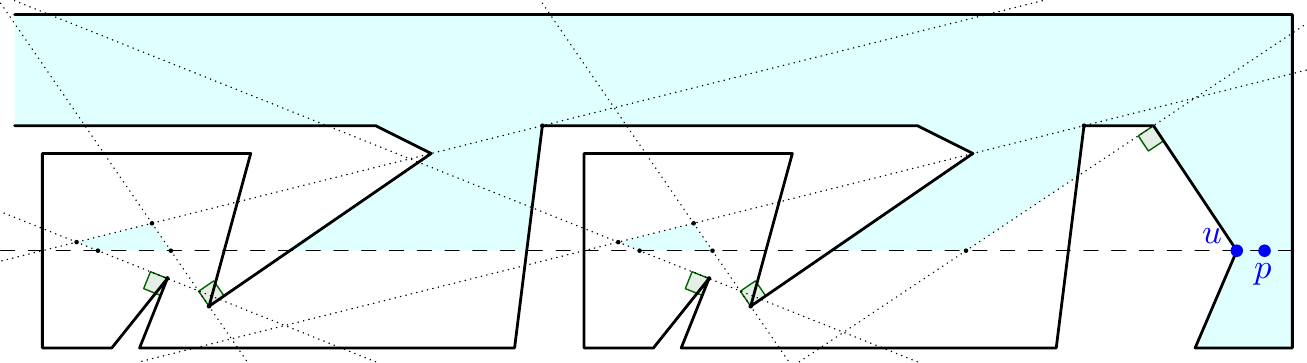}
\caption{\label{fig54} A constraining half-plane may contribute  $O(n)$  vertices of group 2 to the inverse attraction region. Here the inverse attraction region of $p$ is colored.}
\end{figure}

\noindent Note that, as illustrated in Figure~\ref{fig54}, a constraining half-plane may contribute many vertices of group 2 to the inverse attraction region, but nevertheless it is charged at most twice.

\section {Computing the inverse attraction region}
In this section we show how to compute the inverse attraction region of a point inside a simple polygon in $O(n \log n)$ time.

\begin{figure}[b]
\centering
\includegraphics{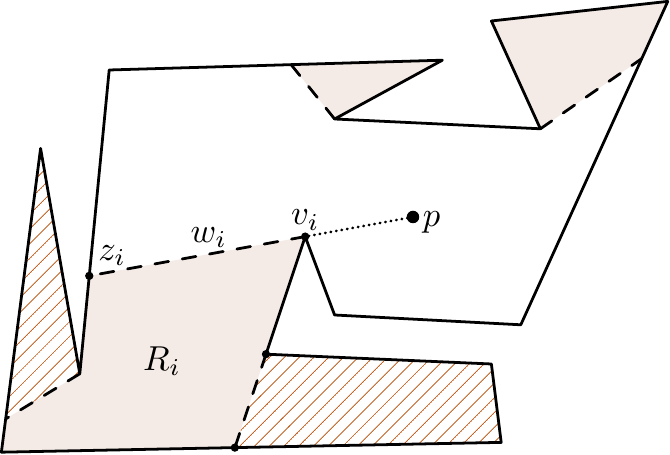}
\caption{\label{fig43} $R_{i}$ is a region of $\spm(p)$ with base $v_{i}$. Segment $w_i$ is the window, and $z_{i}$---its end.}
\end{figure}

Let region $R_{i}$ of the shortest path map $\spm(p)$ consist of all points $t$ such that the last segment of the shortest path from $p$ to $t$ is $\overline{v_i t}$ (Figure~\ref{fig43}). Vertex $v_i$ is called the \emph{base} of $R_{i}$. Extend the edge of $\spt(p)$ ending at $v_i$ until the first intersection $z_i$ with the boundary of $P$. Call the segment $w_i=\overline{v_i z_i}$ a $\emph{window}$, and point $z_i$---the \emph{end} of the window; window $w_i$ is a boundary segment of $R_{i}$.

We will construct a part of the inverse attraction region of $p$ inside each region of the shortest path map $\spm(p)$ independently. A point in a region of $\spm(p)$ attracts $p$ only if its attraction can move $p$ into the region through the corresponding window.

\begin{lemma}
\label{lem:10}
Let $R_{i}$ be a region of $\spm(p)$ with a base vertex $v_i$. If $v_i$ lies in some domain subpolygon $P_e$, then any point $t$ in $R_{i}$ lies in $P_e$.
\end{lemma}
\begin{proof}
Observe, that a shortest path between two points inside a polygon can cross a segment connecting two boundary vertices of $P$ visible to each other at most once.

Let the subpolygon $P_e$ be induced by a segment $\overline{v_j z}$, where $v_j\not=v_i$. If $v_i$ lies inside $P_e$, then the shortest path from $v_i$ to $p$ intersects $\overline{v_j z}$, and the intersection point is not $v_i$. Segment $\overline{tv_i}$ cannot intersect $\overline{v_j z}$, otherwise the shortest path from $t$ to $p$ would cross $\overline{v_j z}$ more than once.

Now let the subpolygon $P_e$ be induced by a segment $\overline{v_i z}$, and let $u_i=\parent(v_i)$. 
Then, $\overline{v_i z}$ is either perpendicular to  $\overline{u_i v_i}$
(Case 1 of Section~\ref{sec:3.1}), or the extension of the edge ``facing'' $u_i$ (Case 2 of Section~\ref{sec:3.1}). 
In either cases $t$ lies inside $P_e$.
\end{proof}

Let $R_{i}$ be a region of $\spm(p)$ with a base vertex $v_i$, and let $\H_i$ be the set of all constraining half-planes corresponding to the domain subpolygons that contain the point $v_i$. Denote $\free_{i}$ to be the intersection of the complements of the half-planes in $\H_i$. Note, that $\free_{i}$ is a convex set. In the following lemma we show that $\free_{i} \cap R_{i}$ is exactly the set of points inside $R_{i}$ that can attract $p$.

\begin{lemma}
\label{lem:11}
The set of points in $R_{i}$ that attract $p$ is $\free_{i} \cap R_{i}$.
\end{lemma}
\begin{proof}
Consider a point $t$ in $R_{i}$. If $t$ lies in a constraining region of one of the domain subpolygons containing $v_i$ (and thus $t$ does not attract $p$), then $t\not\in\free_{i}$, and thus $t\not\in\free_{i} \cap R_{i}$.

If $t\in\free_{i} \cap R_{i}$, then $t$ does not lie in any of the constraining regions of the domain subpolygons containing $v_i$. Assume that $t$ does not attract $p$, \ie there is a separation edge $s$ of $\ar(t)$, such that $p$ and $t$ are in the different subpolygons induced by $s$. Let $v_j\not=v_i$ be the reflex vertex that introduces $s$. Then $t$ does not see vertex $v_j$. Otherwise, as $p$ and $t$ lie in the different subpolygons induced by $s$, and $s$ and $t$ are collinear, vertex $v_j$ would be the base vertex of a region of $\spm(p)$ containing $t$. As $v_i\in\ar(t)$, points $t$ and $v_i$ are in the same subpolygon induced by $s$. Then the domain subpolygon of the constraining half-plane of $v_j$ either contains both $v_i$ and $p$, or neither. Thus, if $t$ does not attract $p$, then it cannot lie in $\free_{i}$.
\end{proof}

This results in the following algorithm for computing the inverse attraction region of $p$. We compute the constraining half-planes of every edge of $\spt(p)$ of $p$ and the corresponding domain subpolygons. Then, for every region $R_{i}$ of the shortest path map of $p$, we compute the free region $\free_{i}$, where $v_i$ is the base vertex of the region; and we add the intersection of $R_{i}$ and $\free_{i}$ to the inverse attraction region of $p$. The pseudocode is presented in Algorithm~\ref{alg:1}.

\begin{algorithm}[b]
\caption{Inverse attraction region.}\label{alg:1}
\renewcommand{\algorithmicrequire}{\textbf{Input:}}
\renewcommand{\algorithmicensure}{\textbf{Output:}}
\renewcommand{\algorithmicforall}{\textbf{for each}}
\begin{algorithmic}[1]
\REQUIRE Simple polygon $P$, and a point $p \in P$.
\ENSURE Inverse attraction region of $p$.
\STATE Compute $\spt(p)$ and $\spm(p)$. 
\FORALL{edge $e \in\spt(p)$}
\STATE Compute constraining half-planes of $e$ and corresponding domain subpolygons.
\ENDFOR
\FORALL{region $R_i$ of $\spm(p)$ with base vertex $v_i$}
\STATE Find all the domain subpolygons that contain $v_i$, and compute $\free_{i}$.
\STATE Intersect $R$ with $\free_{i}$, and add the resulting set to the inverse attraction region of $p$.
\ENDFOR
\RETURN Inverse attraction region of $p$.
\end{algorithmic}
\end{algorithm}

Rather than computing each free space from scratch, we can compute and update free spaces using the data structure of Brodal and Jacob~\cite{BD02}. 
Their data structure allows to dynamically maintain the convex hull of a set of points and supports insertions and deletions in amortized $O(\log n)$ time using $O(n)$ space. In the dual space this is equivalent to maintaining the intersection of $n$ half-planes. In order to achieve a total $O(n \log n)$ time, we need to provide a way to traverse recursive visibility regions and guarantee that the number of updates (insertions or deletions of half-planes) in the data structure is $O(n)$. In the rest of this section, we provide a proof for the following lemma.
\begin{lemma}
\label{theorem:freeSpaces}
Free spaces of the recursive visibility regions can be computed in a total time of $O(n \log n)$ using $O(n)$ space.
\end{lemma}
\begin{proof}
Consider a region $R_{i}$ of $\spm(p)$ with a base vertex $v_i$. By Lemma~\ref{lem:10} and Theorem~\ref{thm:1}, the set of constraining half-planes that affect the inverse attraction region inside $R_{i}$ corresponds to the domain subpolygons that contain $v_i$.

Observe that the vertices of a domain subpolygon appear as one continuous interval along the boundary of $P$, as there is only one boundary segment of the subpolygon that crosses $P$. Then, when walking along the boundary of $P$, each domain subpolygon can be entered and exited at most once. 
All the domain polygons can be computed in $O(n \log n)$ time by shooting $n$ rays and computing their intersection points with the boundary of $P$~\cite{CG89}.


Let the vertices of $P$ be ordered in the counter-clockwise order. For each domain subpolygon $P_e$, mark the two endpoints (e.g., vertices $v$ and $z$ in Figure~\ref{fig36}) of the boundary edge that crosses $P$ as the first and the last vertices of $P_e$ in accordance to the counter-clockwise order. Then, to obtain the optimal running time, we modify the second for-loop of the Algorithm~\ref{alg:1} in the following way. Start at any vertex $v_0$ of $P$, find all the domain subpolygons that contain $v_0$, and initialize the dynamic convex hull data structure of Brodal and Jacob~\cite{BD02} with the points dual to the lines supporting the constraining half-planes of the corresponding domain subpolygons. If $v_0$ is a base vertex of some region $R_{0}$ of $\spm(p)$, then compute the intersection of $R_{0}$ and the free space $\free_{0})$ that we obtain from the dynamic convex hull data structure. Walk along the boundary of $P$ in the counter-clockwise direction, adding to the data structure the dual points to the supporting lines of domain polygons being entered, removing from the data structure the dual points to the supporting lines of domain polygons being exited, and computing the intersection of each region of $\spm(p)$ with the free space obtained from the data structure.

The correctness of the algorithm follows from Lemma~\ref{lem:11}, and the total running time is $O(n \log n)$. Indeed, there will be $O(n)$ updates to the dynamic convex hull data structure, each requiring $O(\log n)$ amortized time. Intersecting free spaces with regions of $\spm(p)$ will take $O(n \log n)$ time in total, as the complexity of $\iar(p)$ is linear. The pseudocode of the algorithm is presented in Appendix~\ref{appendix:pseudocode}.
\end{proof}

\subsection{Lower Bound}
The proof of the following theorem is based on a reduction from the problem of computing the lower envelope of a set of lines, which has a lower bound of $\Omega(n \log n)$~\cite{Y81}. 
\begin{theorem}
\label{thm:reduction}
Computing the inverse attraction region of a point in a monotone (or a simple polygon) has a lower bound of $\Omega(n \log n)$. 
\end{theorem}
\begin{proof}
Consider a set of lines $L$. Let $l_b$ and $l_s$ denote the lines in $L$ with the biggest and smallest slope, respectively. Note that the leftmost (rightmost) edge of the lower envelope of $L$ is part of $l_b$ ($l_s$).

Without loss of generality assume that the slopes of the lines in $L$ are positive and bounden from above by a small constant $\varepsilon$. 
We construct a monotone polygon as follows.
The right part of the polygon is comprised of an axis aligned rectangle $R$ that contains all the intersection points of the lines in $L$ 
(Figure~\ref{fig35}). Note that $R$ can be computed in linear time.

\begin{figure}[t]
\begin{minipage}[t]{0.48\textwidth}
\centering
\includegraphics[width=\textwidth]{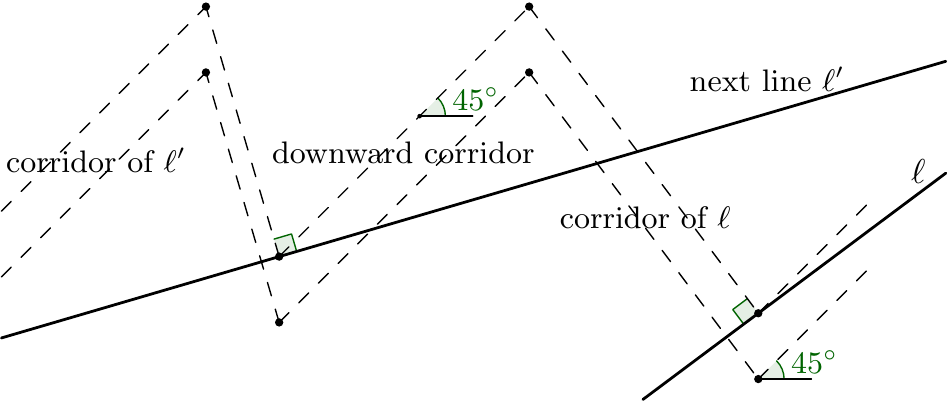}
\caption{\label{fig32}Adding a corridor for a line of $L$.}
\end{minipage}
\hfill
\begin{minipage}[t]{0.5\textwidth}
\centering
\includegraphics[width=\textwidth]{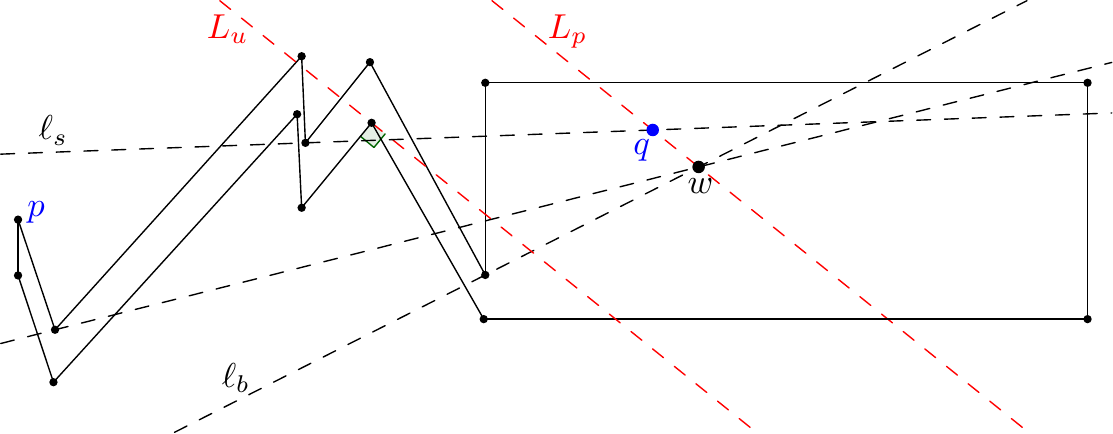}
\caption{\label{fig35}The final monotone polygon constructed for 3 lines.}
\end{minipage}
\end{figure}

To the left of $R$, we construct a ``zigzag'' corridor in the following way. For each line $l$ in $L$, in an arbitrary order, we add a corridor perpendicular to $l$ which extends above the next arbitrarily chosen line (Figure~\ref{fig32}). We then add a corridor with slope $1$ going downward until it hits the next line. This process is continued for all lines in $L$. 

Let the point $p$ be the leftmost vertex of the upper chain of the corridor structure. Consider the inverse attraction region of $p$ in the resulting monotone polygon. 
A point in $R$ can attract $p$, only if it is below all lines of $L$, \ie only if it is below the lower envelope of $L$. In addition the point needs to be above the line $L_u$, where $L_u$ is the rightmost line perpendicular to a lower edge of the corridors with a slope of $-1$ (refer to Figure~\ref{fig35}). 
In order to have all vertices of the lower envelope in the inverse attraction region, we need to guarantee that $L_u$ is to the left of the leftmost vertex of the lower envelope, $w$. Let $L_p$ be a line through $w$ with a scope equal to $-1$. Let $q$ be the intersection of $L_p$ with $l_s$. We start the first corridor of the zigzag to the left of $q$. As the lines have similar slopes this guarantees that $L_u$ is to the left of vertices of the lower envelope. 
Now it is straightforward to compute the lower envelope of $L$ in linear time given the inverse attraction region of $p$.  
\end{proof}

We conclude with the main result of this paper.
\begin{theorem}
The inverse attraction region of a point in a simple polygon can be computed in $\Theta(n \log n)$ time. 
\end{theorem}
%
\bibliographystyle{abbrv}
\bibliography{refs}

\begin{thebibliography}{10}

\bibitem{BSV}
S.~W. Bae, C.-S. Shin, and A.~Vigneron.
\newblock Tight bounds for beacon-based coverage in simple rectilinear
  polygons.
\newblock In {\em 12th Latin American Symposium on Theoretical Informatics},
  2016.

\bibitem{B13}
M.~Biro.
\newblock {\em Beacon-based routing and guarding}.
\newblock PhD thesis, Stony Brook University, 2013.

\bibitem{BGIKM}
M.~Biro, J.~Gao, J.~Iwerks, I.~Kostitsyna, and J.~S.~B. Mitchell.
\newblock Beacon-based routing and coverage.
\newblock In {\em 21st Fall Workshop on Computational Geometry}, 2011.

\bibitem{BGIKM2}
M.~Biro, J.~Gao, J.~Iwerks, I.~Kostitsyna, and J.~S.~B. Mitchell.
\newblock Combinatorics of beacon-based routing and coverage.
\newblock In {\em 25th Canadian Conference on Computational Geometry}, 2013.

\bibitem{BIKM}
M.~Biro, J.~Iwerks, I.~Kostitsyna, and J.~S.~B. Mitchell.
\newblock Beacon-based algorithms for geometric routing.
\newblock In {\em 13th Algorithms and Data Structures Symposium}, 2013.

\bibitem{BD02}
G.~S. Brodal and R.~Jacob.
\newblock Dynamic planar convex hull.
\newblock In {\em 43rd Annual IEEE Symposium on Foundations of Computer
  Science}, 2002.

\bibitem{CG89}
B.~Chazelle and L.~J. Guibas.
\newblock Visibility and intersection problems in plane geometry.
\newblock {\em Discrete {\&} Computational Geometry}, 4(6):551--581, 1989.

\bibitem{GHLST}
L.~J. Guibas, J.~Hershberger, D.~Leven, M.~Sharir, and R.~Tarjan.
\newblock Linear-time algorithms for visibility and shortest path problems
  inside triangulated simple polygons.
\newblock {\em Algorithmica}, 2(1--4):209--233, 1987.

\bibitem{KT10}
A.-M. Kermarrec and G.~Tan.
\newblock Greedy geographic routing in large-scale sensor networks: A minimum
  network decomposition approach.
\newblock {\em IEEE/ACM Transactions on Networking}, 20:864--877, 2010.

\bibitem{KRS15b}
B.~Kouhestani, D.~Rappaport, and K.~Salomaa.
\newblock The length of the beacon attraction trajectory.
\newblock In {\em 27th Canadian Conference on Computational Geometry}, 2015.

\bibitem{KRS15}
B.~Kouhestani, D.~Rappaport, and K.~Salomaa.
\newblock On the inverse beacon attraction region of a point.
\newblock In {\em 27th Canadian Conference on Computational Geometry}, 2015.

\bibitem{KRS14}
B.~Kouhestani, D.~Rappaport, and K.~Salomaa.
\newblock Routing in a polygonal terrain with the shortest beacon watchtower.
\newblock {\em International Journal of Computational Geometry \&
  Applications}, 68:34--47, 2018.

\bibitem{LP}
D.~T. Lee and F.~P. Preparata.
\newblock Euclidean shortest paths in the presence of rectilinear barriers.
\newblock {\em Networks}, 14(3):393--410, 1984.

\bibitem{S15}
T.~Shermer.
\newblock A combinatorial bound for beacon-based routing in orthogonal
  polygons.
\newblock In {\em 27th Canadian Conference on Computational Geometry}, 2015.

\bibitem{Y81}
A.~C. Yao.
\newblock A lower bound to finding convex hulls.
\newblock {\em Journal of the ACM}, 28(4):780--787, 1981.

\end{thebibliography}

\clearpage
\appendix

\section{Pseudocode for computing \texorpdfstring{$\iar(p)$}{IAR(p)}}\label{appendix:pseudocode}

\begin{algorithm}[h]
\caption{Optimal inverse attraction region computation.}\label{alg:2}
\renewcommand{\algorithmicrequire}{\textbf{Input:}}
\renewcommand{\algorithmicensure}{\textbf{Output:}}
\renewcommand{\algorithmicforall}{\textbf{for each}}
\begin{algorithmic}[1]
\REQUIRE Simple polygon $P$, and a point $p \in P$.
\ENSURE Inverse attraction region of $p$.
\STATE Compute $\spt(p)$ and $\spm(p)$. 
\FORALL{edge $e \in\spt(p)$}
\STATE Compute constraining half-planes of $e$ and corresponding domain subpolygons.
\ENDFOR
\STATE For some boundary point $v_0$ of $P$, initialize the dynamic convex hull data structure with the points dual to the supporting lines of the constraining half-planes of the domain subpolygons containing $v_0$.
\STATE $v \leftarrow v_0$
\REPEAT
\IF{$v$ is a base vertex of some region $R$ of $\spm(p)$}
\STATE Get region $\free$ from the dynamic convex hull data structure.
\STATE Add $R\cap\free$ to the inverse attraction region of $p$.
\ENDIF
\STATE $v \leftarrow v.\mathit{next}()$ \qquad\quad\COMMENT{Next boundary vertex of $P$ in the counter-clockwise direction}
\FORALL{subpolygon $P_e$ entered}
\STATE Insert the point dual to the supporting line of the constraining half-plane of $P_e$ into the dynamic convex hull data structure.
\ENDFOR
\FORALL{subpoygon $P_e$ exited}
\STATE Delete the point dual to the supporting line of the constraining half-plane of $P_e$ from the dynamic convex hull data structure.
\ENDFOR
\UNTIL{$v\not=v_0$}
\RETURN Inverse attraction region of $p$.
\end{algorithmic}
\end{algorithm}

\end{document}